\documentclass[11pt]{article}

\usepackage[utf8]{inputenc}
\usepackage[english]{babel}
\usepackage{graphicx}         
\usepackage[pdftex]{hyperref}
\usepackage{braket}
\usepackage{amsmath}
\usepackage{amssymb}
\usepackage{amsthm}
\usepackage{mathtools}
\usepackage{bbm}
\usepackage{array}		
\usepackage{color}
\usepackage[usenames,dvipsnames]{xcolor}
\usepackage{caption}
\usepackage{enumitem}
\usepackage{subeqnarray}	
\usepackage{wrapfig}
\usepackage{cancel}
\usepackage{enumitem}

\allowdisplaybreaks
\newtheorem{theorem}{Theorem}
\newtheorem{assumption}[theorem]{Assumption}

\newtheorem*{conjecture*}{Conjecture}

\newtheorem{proposition}[theorem]{Proposition}

\newtheorem{remark}[theorem]{Remark}

\newcommand{\D}{\mathcal{D}}
\newcommand{\R}{\mathbb{R}}

\newcommand{\N}{\mathbb{N}}
\newcommand{\Z}{\mathbb{Z}}
\newcommand{\T}{\mathbb{T}}

\newcommand{\Fock}{\mathcal{F}}
\newcommand{\Number}{\mathcal{N}}

\newcommand{\cQ}{\mathcal{Q}}

\renewcommand{\Re}{\mathrm{Re}}

\newcommand{\FockO}{\mathbb{O}}
\newcommand{\FockP}{\mathbb{P}}
\newcommand{\FockQ}{\mathbb{Q}}

\newcommand{\hc}{\mathrm{h.c.}}

\newcommand{\sym}{\mathrm{sym}}
\newcommand{\Tr}{\mathrm{Tr}}

\newcommand{\exc}{\mathrm{exc}}
\newcommand{\extra}{\mathrm{extra}}

\newcommand{\binding}{\mathrm{binding}}

\newcommand{\bj}{\boldsymbol{j}}

\newcommand{\bm}{\boldsymbol{m}}

\renewcommand{\hat}[1]{\widehat{#1}}

\newcommand{\scp}[2]{\langle #1, #2 \rangle}
\newcommand{\SCP}[2]{\left\langle #1, #2 \right\rangle}

\renewcommand{\d}{\mathop{}\!\mathrm{d}}
\newcommand{\dx}{\d x}
\newcommand{\dy}{\d y}

\newcommand{\dGamma}{\mathop{}\!\mathrm{d} \Gamma}

\newcommand{\ad}{a^*}

\newcommand\mydots{,\makebox[1em][c]{.\hfil.\hfil.},}
\newcommand\mycdots{\makebox[1em][c]{$\cdot$\hfil$\cdot$\hfil$\cdot$}}

\makeatletter
\DeclareFontFamily{OMX}{MnSymbolE}{}
\DeclareSymbolFont{MnLargeSymbols}{OMX}{MnSymbolE}{m}{n}
\SetSymbolFont{MnLargeSymbols}{bold}{OMX}{MnSymbolE}{b}{n}
\DeclareFontShape{OMX}{MnSymbolE}{m}{n}{
    <-6>  MnSymbolE5
   <6-7>  MnSymbolE6
   <7-8>  MnSymbolE7
   <8-9>  MnSymbolE8
   <9-10> MnSymbolE9
  <10-12> MnSymbolE10
  <12->   MnSymbolE12
}{}
\DeclareFontShape{OMX}{MnSymbolE}{b}{n}{
    <-6>  MnSymbolE-Bold5
   <6-7>  MnSymbolE-Bold6
   <7-8>  MnSymbolE-Bold7
   <8-9>  MnSymbolE-Bold8
   <9-10> MnSymbolE-Bold9
  <10-12> MnSymbolE-Bold10
  <12->   MnSymbolE-Bold12
}{}
\let\llangle\@undefined
\let\rrangle\@undefined
\DeclareMathDelimiter{\llangle}{\mathopen}
                     {MnLargeSymbols}{'164}{MnLargeSymbols}{'164}
\DeclareMathDelimiter{\rrangle}{\mathclose}
                     {MnLargeSymbols}{'171}{MnLargeSymbols}{'171}
\makeatother

\newcommand{\Vext}{V^\mathrm{ext}}

\newcommand{\hH}{h}
\newcommand{\eH}{e_\mathrm{H}}

\newcommand{\UNp}{U_{N,\varphi}}
\newcommand{\UNmop}{U_{N-1,\varphi}}

\newcommand{\FockH}{\mathbb{H}}

\newcommand{\tildeFockH}{\widetilde{\mathbb{H}}}

\newcommand{\QP}{\mathrm{QP}}
\newcommand{\FockHz}{\FockH_0}

\newcommand{\FockHjo}{\FockH_{j_1}}

\newcommand{\FockHjnu}{\FockH_{j_\nu}}

\newcommand{\FockR}{\mathbb{R}}

\newcommand{\Chi}{{\boldsymbol{\chi}}}

\renewcommand{\P}{\mathbb{P}}

\newcommand{\Q}{\mathbb{Q}}

\newcommand{\Pz}{\P_0}

\newcommand{\Qz}{\Q_0}

\newcommand{\boldK}{\mathbb{K}}
\newcommand{\boldKz}{\mathbb{K}_0}
\newcommand{\boldKo}{\mathbb{K}_1}
\newcommand{\boldKt}{\mathbb{K}_2}
\newcommand{\boldKth}{\mathbb{K}_3}
\newcommand{\boldKf}{\mathbb{K}_4}

\title{A Note on the Binding Energy for Bosons in the \\ Mean-field Limit}

\author{Lea Boßmann\thanks{Institut für Mathematik, Johannes Gutenberg-Universität Mainz, Staudingerweg 9, 55128 Mainz, Germany; and Department of Mathematics, Ludwig-Maximilians-Universität München, Theresienstr.\ 39, 80333 München, Germany. Email: \texttt{bossmann@uni-mainz.de}},
Nikolai Leopold\thanks{Department of Mathematics and Computer Science, University of Basel, Spiegelgasse 1, 4051 Basel,	Switzerland. Email: \texttt{nikolai.leopold@unibas.ch}},
David Mitrouskas\thanks{Institute of Science and Technology Austria (ISTA), Am Campus 1, 3400 Klosterneuburg, Austria. Email: \texttt{david.mitrouskas@ist.ac.at}}, and 
Sören Petrat\thanks{School of Science, Constructor University Bremen, Campus Ring 1, 28759 Bremen, Germany. Email: \texttt{spetrat@constructor.university}}}

\date{\today}
\begin{document}
\maketitle

\begin{abstract}
We consider a gas of $N$ weakly interacting bosons in the ground state. Such gases exhibit Bose--Einstein condensation. The binding energy is defined as the energy it takes to remove one particle from the gas. In this article, we prove an asymptotic expansion for the binding energy, and compute the first orders explicitly for the homogeneous gas. Our result addresses in particular a conjecture by Nam [Lett.\ Math.\ Phys., 108(1):141--159, 2018], and provides an asymptotic expansion of the ionization energy of bosonic atoms.
\end{abstract}

\section{Introduction and Main Results}

We consider the $N$-particle Hamiltonian
\begin{equation}\label{H_N_w}
H(N,w) = \sum_{i=1}^N T_i + \sum_{1 \leq i < j \leq N} w(x_i-x_j),
\end{equation}
with $T_i = -\Delta_i + \Vext(x_i)$ and $\Vext,w:\R^d\to\R$, as an operator on the bosonic Hilbert space
\begin{equation}
\mathcal{H}^N_{\sym} = \big( L^2(\Omega) \big)^{\otimes_s N},
\end{equation}
with $\otimes_s$ the symmetric tensor product. We let the dimension $d \geq 1$ and distinguish two cases:
\begin{itemize}
\item $\Omega = \R^d$, with $\Vext(x) \to \infty$ as $|x| \to \infty$. In this case we call the system the \textbf{trapped Bose gas}.
\item $\Omega = \mathbb{T}^d$, the unit torus. In this case we set $\Vext = 0$ and call the system the \textbf{homogeneous Bose gas}.
\end{itemize}
We are interested in the mean-field limit, i.e., an interaction 
\begin{equation}
w = \lambda_N v, ~~\text{with}~\lambda_N := (N-1)^{-1},
\end{equation}
and $v:\R^d\to\R$. The spectral properties of \eqref{H_N_w} in the mean-field limit have been extensively studied; let us refer to \cite{LSSY} for a more general review of the mean-field and more singular models. The leading order of the energy is described by the Hartree energy functional \eqref{Hartree_functional}. More recently, the next-to leading order of the low-lying eigenvalues and the corresponding eigenfunctions has been understood rigorously in terms of Bogoliubov theory, see \cite{seiringer2011,grech2013,lewin2014,lewin2015_2,mitrouskas_PhD} for recent results, and \cite{bogoliubov1947} for Bogoliubov's original paper. The eigenfunctions in Bogoliubov theory are described in terms of quasi-free states (and the ground state is exactly a quasi-free state). This allows in particular a perturbative expansion around Bogoliubov theory with coefficients that can be explicitly computed, see \cite{spectrum}. In this article we explore the consequences of this perturbative expansion in more detail by proving an expansion not just for the energy of an $N$-body system, but for the binding energy. If the many-body system is an atom, this quantity is known as the ionization energy.

Let us denote the ground state energy, i.e., the lowest eigenvalue of $H(N,w)$, by $E(N,w)$. The binding energy is the energy necessary to remove one particle from the ground state, i.e., it is defined as
\begin{equation}\label{def_binding_energy}
\Delta E (N,\lambda_N v) := E(N,\lambda_N v) - E(N-1,\lambda_N v).
\end{equation}
Here, we assume that the Bose gases of $N$ and $N-1$ particles have the same coupling constant $\lambda_N$. In \cite{nam2018}, it was proven by Nam that for the homogeneous Bose gas
\begin{align}\label{Nam_result}
\Delta E (N,\lambda v) = \lambda (N-1) \hat{v}(0) + \frac{1}{N} \Bigg( e_B - \sum_{\substack{p \in (2\pi \Z)^d \\ p \neq 0}} \frac{p^2 \alpha_p^2}{1-\alpha_p^2} + o(1) \Bigg)
\end{align}
in the limit $N\to \infty$ and $\lambda N \to 1$, where
\begin{align}\label{alpha_p}
\alpha_p := \frac{\hat{v}(p)}{p^2 + \hat{v}(p) + \sqrt{p^4 + 2p^2 \hat{v}(p)}}, ~~ e_B := -\frac{1}{2} \sum_{\substack{p \in (2\pi \Z)^d \\ p \neq 0}} \alpha_p \hat{v}(p),
\end{align}
and $\hat{v}(p) := \int v(x) e^{-ipx} \dx$ for all $p \in (2\pi \Z)^d$ denotes the Fourier transform of $v$. The result holds for even and bounded $v$ with nonnegative Fourier transform. We improve this result in two directions:
\begin{itemize}
\item We prove an asymptotic expansion of $\Delta E (N,\lambda_N v)$ in powers of $\lambda_N$.
\item We prove this expansion for both the homogeneous and the trapped Bose gas.
\end{itemize}
Note that Nam mentioned an extension of \eqref{Nam_result} to trapped bosons as an open problem and set up a conjecture about this generalization, see \cite[Conjecture~6]{nam2018}. We address this problem in particular with Theorem~\ref{theorem_main} and elaborate on the conjecture in Remark~\ref{remark_Nam_conjecture} and Section~\ref{sec_next_order_Nam}.

The proof of an asymptotic expansion of the binding energy has become possible through the work \cite{spectrum}, where asymptotic expansions for the ground state, low energy excited states, and their corresponding energies have been proven. Our article is an application of that expansion for the ground state energy. Note that the work \cite{spectrum} was in turn inspired by an analogous result for the dynamics \cite{QF}; see also the follow-up work \cite{FLMP2021}. Let us refer to \cite{proceedings} and \cite{detlef_proceedings} for reviews of both results, and note that in \cite{CLT} the results from \cite{spectrum} are applied to derive an Edgeworth expansion for the fluctuations of bounded one-body operators with respect to the ground state and low-energy excited states of the weakly interacting Bose gas.

In order to state our main results we need a few technical assumptions. These are the same assumptions that were made for proving the asymptotic expansion of the ground state and the ground state energy in \cite{spectrum}. We briefly list and explain these assumptions here and refer to \cite[Section~2.1]{spectrum} for more details.

\begin{assumption}\label{Assumption_1}
Let $\Vext:\R^d\to\R$ be measurable, locally bounded and non-negative and let $\Vext(x)$ tend to infinity as $|x|\to  \infty$, i.e.,
\begin{equation}
\inf\limits_{|x|>R}\Vext(x)\to\infty \text{ as } R\to \infty\,.
\end{equation}
\end{assumption}
This assumption implies in particular that $\Vext$ is confining.

\begin{assumption}\label{Assumption_2}
Let $v:\R^d\to\R$ be measurable with $v(-x)=v(x)$ and $v\not\equiv 0$, and assume that there exists a constant $C>0$ such that, in the sense of operators on $\cQ(-\Delta)=H^1(\R^d)$,
\begin{equation}
|v|^2\leq  C\left(1-\Delta\right)\,.  \label{eqn:ass:v:2:Delta:bound}
\end{equation}
Besides, assume that $v$ is of positive type, i.e., that it has a non-negative Fourier transform.
\end{assumption}
Together, Assumptions \ref{Assumption_1} and \ref{Assumption_2} imply self-adjointness of $H(N,\lambda v)$ for any $\lambda \in \R$ (by Kato--Rellich). Let us recall that it has been proven in many settings that weakly interacting bosons exhibit Bose--Einstein condensation, which means a macroscopic occupation of the one-particle state $\varphi \in L^2(\Omega)$. In our setting the condensate wave function $\varphi$ is the minimizer of the Hartree energy functional
\begin{equation}\label{Hartree_functional}
\mathcal{E}_{\mathrm{H}}[\phi] := \int \Big( |\nabla \phi(x)|^2 + \Vext(x) |\phi(x)|^2 \Big) \dx + \frac{1}{2} \int v(x-y) |\phi(x)|^2 |\phi(y)|^2 \dx \dy.
\end{equation}
The corresponding Hartree energy is $\eH := \inf_{\phi \in H^1(\Omega), \|\phi\|=1} \mathcal{E}_{\mathrm{H}}[\phi] = \mathcal{E}_{\mathrm{H}}[\varphi]$. Assumptions \ref{Assumption_1} and \ref{Assumption_2} imply all necessary properties of the Hartree minimizer $\varphi$, in particular its existence and uniqueness, and the existence of a spectral gap above the ground state of the one-body Hartree operator $h = T + v*|\varphi|^2$.

\begin{assumption}\label{Assumption_3}
Assume that there exist constants $C_1\geq0$ and $0<C_2\leq 1$, as well as a function $\varepsilon:\N\to\R_0^+$ with 
$$\lim\limits_{N\to\infty} N^{-\frac13}\varepsilon(N) \leq C_1\,,$$
such that 
\begin{equation}\label{eqn:ass:cond}
H(N,\lambda_N v)-N\eH\geq C_2 \sum\limits_{j=1}^N\hH_j-\varepsilon(N)
\end{equation}
in the sense of operators on $\D(H(N,\lambda_N v))$.
\end{assumption}
Assumptions~\ref{Assumption_2} and \ref{Assumption_3} hold in particular for any bounded even $v$ with nonnegative Fourier transform \cite{grech2013}, and for the three-dimensional repulsive Coulomb potential $v(x) = |x|^{-1}$ \cite{lewin2015_2}. Assumption~\ref{Assumption_3} ensures complete Bose--Einstein condensation of the $N$-body state in the Hartree minimizer $\varphi$ with a sufficiently good rate. With these assumptions we can state our main results.

\begin{theorem}\label{theorem_main}
Consider the trapped Bose gas, i.e., the Hamiltonian
\begin{equation}\label{H_N_theorem}
H(N,\lambda_Nv ) = \sum_{i=1}^N \big(-\Delta_i + \Vext(x_i)\big) + \lambda_N \sum_{1 \leq i < j \leq N} v(x_i-x_j),
\end{equation}
and let Assumptions~\ref{Assumption_1}, \ref{Assumption_2}, and \ref{Assumption_3} hold. Then, for any $a \in \N$, the binding energy as defined in \eqref{def_binding_energy} has an expansion
\begin{equation}\label{main_expansion}
\Delta E (N,\lambda_N v) =  \sum_{j=0}^a \lambda_N^{j} E^{\mathrm{binding}}_j + O(\lambda_N^{a+1}).
\end{equation}
We have
\begin{equation}
E_0^\binding = \eH + \frac{1}{2} \scp{\varphi}{\big(v*|\varphi|^2\big) \varphi} = \scp{\varphi}{\big(-\Delta + \Vext + v*|\varphi|^2 \big) \varphi},
\end{equation}
and
the coefficients $E^{\mathrm{binding}}_j$ for $j\geq 1$ are stated in Proposition~\ref{lemma_energy_expansions}. 
\end{theorem}

\begin{proof}
The theorem follows from the corresponding expansions for $E(N,\lambda_N v)$ and $E(N-1,\lambda_N v)$ in Proposition~\ref{lemma_energy_expansions}.
\end{proof}

\begin{remark}\label{remark_Nam_conjecture}
Let us compare this result with \cite[Conjecture~6]{nam2018}. Note that here we have adapted the conjecture to our notation.
\begin{conjecture*}[{\cite[Conjecture~6]{nam2018}}]\label{Nams_conjecture}
Under appropriate conditions on $T$ and $v$,
\begin{equation}
E(N,\lambda v) - E(N-1,\lambda v) = A + CN^{-1} + o(N^{-1})
\end{equation}
as $N \to \infty$ and $\lambda N \to 1$, with coefficients $A$ and $C$ as given in \cite[Section~5]{nam2018} (or see Section~\ref{sec_next_order_Nam}).
\end{conjecture*}
In particular, $A = E_0^\binding$. However, the conjectured coefficient $C$ is in general not equal to $E_1^\binding$, except for the homogeneous Bose gas. We elaborate on this in Section~\ref{sec_next_order_Nam}.
\end{remark}

\begin{remark} Note that Theorem \ref{theorem_main} also applies to bosonic atoms, where the binding energy is referred to as ionization energy \cite{Bach1991}. An atom with $N$ spinless ``bosonic electrons'' and $Z$ nuclei is described by the Hamiltonian
\begin{align}
H^{\rm atom}_{N,Z} = \sum_{i=1}^N \bigg( -\Delta_i - \frac{Z}{|x_i|} \bigg) +   \sum_{1 \le i<j \le N } \frac{1}{|x_i-x_j|},
\end{align}
acting on $\mathcal{H}^N_{\sym}$. Rescaling the coordinates $x_i \to \lambda_N x_i $ and setting $ t = ( N -1) / Z $ leads to
\begin{align}\label{bosonic_atoms}
\lambda_N^2 H^{\rm atom}_{N,t} =  \sum_{i=1}^N \bigg( -\Delta_i -  \frac{ 1 }{ t |x_i|} \bigg) +  \lambda  _N \sum_{1\le i<j \le N } \frac{1}{|x_i-x_j|}.
\end{align}
We consider the limit where $N\to\infty$ with $t$ fixed. It is known \cite{Lieb1981} that there is a critical $t_c\in (1,2)$ such that for $t\le t_c$, the quantum problem and the corresponding Hartree energy functional have unique ground states. That the first-order contribution of the ground state energy is given by $\inf \sigma (H_{N, t}) = N e_H(t) + o(N) $ as $N\to \infty$, where $e_H(t)$ is the infimum of the corresponding Hartree energy functional, was proved by Benguria and Lieb \cite{benguria-lieb1983}. Bach \cite{Bach1991} showed that the first-order contribution to the ionization energy can be described as well in terms of the Hartree energy. In \cite{lewin2015_2,nam_PhD}, it was then shown that the low-energy eigenvalues of $H_{N, t}$ below the essential spectrum are determined by Bogoliubov theory. As explained in \cite[Remark 3.6]{spectrum} the bosonic atom meets all the required criteria for an asymptotic expansion of the low-energy eigenvalues in inverse powers of $\lambda_N$, similarly as in the case of confined bosons. Since the proof of Theorem \ref{theorem_main} is entirely based on the asymptotic expansion of the low-energy eigenvalues, it also applies to the Hamiltonian \eqref{bosonic_atoms} for bosonic atoms, and thus provides an asymptotic expansion for the ionization energy.
\end{remark}

\begin{remark}
Just as the results of \cite{spectrum}, Theorem~\ref{theorem_main} holds under more general assumptions than Assumptions~\ref{Assumption_1}, \ref{Assumption_2}, and \ref{Assumption_3}. These are the assumptions (A1) and (A2) in \cite{lewin2015_2}, our Assumption~\ref{Assumption_3} (which is slightly stronger than (A3s) from \cite{lewin2015_2}), and Inequality~\eqref{eqn:ass:v:2:Delta:bound}. We refer to \cite[Remark 3.6]{spectrum} for more details. These more general assumptions can be satisfied for interactions $v$ that are not of positive type, for example, the two-dimensional Coulomb gas discussed in \cite[Sec.~3.2]{lewin2015_2}, where $v(x) = -\log |x|$.
\end{remark}

For the homogeneous case, $E^{\mathrm{binding}}_0 = \hat{v}(0)$ can be concluded from \cite{seiringer2011}, and 
\begin{equation}\label{E_1_binding_theorem}
E^\binding_1 = e_B - \sum_{\substack{p \in (2\pi \Z)^d \\ p \neq 0}} \frac{p^2 \alpha_p^2}{1-\alpha_p^2} = - \sum_{\substack{p \in (2\pi \Z)^d \\ p \neq 0}} \hat{v}(p) \frac{\alpha_p}{1+\alpha_p}
\end{equation}
is already known from \cite{nam2018}. We compute here the next coefficient $E^{\mathrm{binding}}_2$. In the following theorem all summations are over the lattice $(2\pi \Z)^d$.

\begin{theorem}\label{theorem_homogeneous}
For the homogeneous Bose gas the expansion \eqref{main_expansion} from Theorem~\ref{theorem_main} is true under Assumption~\ref{Assumption_2} with $\R^d$ replaced by $\T^d$ and Assumption~\ref{Assumption_3}. The second-order coefficient is given by
\begin{align}\label{E_2_binding_theorem}
\begin{split}
E_2^\binding &= \sum_{k \neq 0} \frac{k^2 \gamma_k \sigma_k}{\varepsilon(k)} \Big( k^2 \gamma_k \sigma_k - f(k) \Big) + 6 \sum_{\substack{k,\ell \neq 0 \\ k + \ell \neq 0}} \left(\frac{(k+\ell)^2 g_2(k,\ell)}{\varepsilon(k) + \varepsilon(\ell) + \varepsilon(k+\ell)}\right) \times \\
&\qquad\qquad\qquad\qquad\qquad\qquad \left( \frac{2 \sigma_{k+\ell} \gamma_{k+\ell} g_1(k,\ell)}{\varepsilon(k+\ell)} - \frac{3 \big(\sigma_{k+\ell}^2 + \gamma_{k+\ell}^2\big) g_2(k,\ell)}{\varepsilon(k) + \varepsilon(\ell) + \varepsilon(k+\ell)} \right),
\end{split}
\end{align}
with 
\begin{align}\label{epsilon_alpha_sigma_gamma}
\varepsilon(p) := \sqrt{p^4 + 2p^2\hat{v}(p)}, ~ \alpha_p := \frac{\hat{v}(p)}{p^2 + \hat{v}(p) + \sqrt{p^4 + 2p^2 \hat{v}(p)}}, ~ \sigma_p := \frac{1}{\sqrt{1-\alpha_p^2}}, ~ \gamma_p := \alpha_p \sigma_p,
\end{align}
and
\begin{subequations}
\begin{align}
\begin{split}
f(k) &:= - \sum_{\substack{\ell \neq 0 \\ \ell \neq k}} \hat{v}(k-\ell) \gamma_\ell \Big( \sigma_k^2 \sigma_\ell + 2 \sigma_k \gamma_\ell \gamma_k + \sigma_\ell \gamma_k^2 \Big) - \hat{v}(k) (\sigma_k - \gamma_k)^2 \sum_{\ell \neq 0} \gamma_\ell^2 \\
&\quad - 2 \sigma_k \gamma_k \sum_{\ell \neq 0} \hat{v}(\ell) \gamma_\ell (\sigma_\ell - \gamma_\ell) +  2\hat{v}(k) \gamma_k (\sigma_k - \gamma_k)^3 + \frac{1}{2} \hat{v}(k) \big( \sigma_k^2 + \gamma_k^2 \big),
\end{split} \\
\begin{split} 
g_1(k,\ell) &:= \frac{1}{2} \Bigg[ \hat{v}(k) \big(\sigma_{k+\ell} \sigma_\ell + \gamma_{k+\ell} \gamma_\ell\big)\big(\sigma_k - \gamma_k\big) + \hat{v}(\ell) \big(\sigma_{k+\ell} \sigma_k + \gamma_{k+\ell} \gamma_k\big)\big(\sigma_\ell - \gamma_\ell\big) \\
&\qquad - \hat{v}(k+\ell) \big(\sigma_\ell\gamma_k + \sigma_k\gamma_\ell\big)\big(\sigma_{k+\ell} - \gamma_{k+\ell}\big) \Bigg],
\end{split} \\
\begin{split}
g_2(k,\ell) &:= -\frac{1}{6} \Bigg[ \hat{v}(k) \big(\gamma_{k+\ell} \sigma_\ell + \sigma_{k+\ell} \gamma_\ell\big)\big(\sigma_k - \gamma_k\big) + \hat{v}(\ell) \big(\gamma_{k+\ell} \sigma_k + \sigma_{k+\ell} \gamma_k\big)\big(\sigma_\ell - \gamma_\ell\big) \\
&\qquad + \hat{v}(k+\ell) \big(\sigma_\ell\gamma_k + \sigma_k\gamma_\ell\big)\big(\sigma_{k+\ell} - \gamma_{k+\ell}\big) \Bigg].
\end{split}
\end{align}
\end{subequations}

\end{theorem}

\begin{proof}
The quantity $E^\binding_2$ on the torus is computed in Section~\ref{sec_homogeneous_comp}.
\end{proof}

Note that our analysis can be extended to excited states in a similar way but we do not pursue this here. An interesting open problem would be to prove an expression for the binding energy in the more singular Gross--Pitaevskii regime (see, e.g., \cite{LSSY} and \cite{boccato2018}), where in three dimensions $w(x) = N^2 v(Nx)$ for suitable $N$-independent $v$.

\begin{remark}
Note in particular that $E_0^\binding \geq 0$ and $E_1^\binding \leq 0$. The sign of $E_2^\binding$ is not in general evident. However, for an interaction $\hat{v}_\Lambda(k) := \hat{v}\big(\frac{k}{\Lambda}\big)$ with $\Lambda >0$ large a straightforward computation yields the scaling behavior
\begin{align}
E_2^\binding(\Lambda) = \underbrace{\sum_{k \neq 0} \frac{k^2 \gamma_k(\Lambda) \sigma_k(\Lambda)}{\varepsilon_k(\Lambda)} \sum_{\substack{\ell \neq 0 \\ \ell \neq k}} \hat{v}\Big(\frac{k-\ell}{\Lambda}\Big) \gamma_\ell(\Lambda) \sigma_k(\Lambda)^2 \sigma_\ell(\Lambda)}_{= O(\Lambda^2)} \, + \, O(\Lambda),
\end{align}
and thus we can conclude that $E_2^\binding(\Lambda) \geq 0$ for $\Lambda$ large enough.
\end{remark}

The rest of the article is organized as follows. In Section~\ref{sec_main_theorem_proof}, we prove Proposition~\ref{lemma_energy_expansions} which immediately implies the proof of Theorem~\ref{theorem_main}. More concretely, in Section~\ref{sec_Ham_Exc_Space}, we first conjugate $H(N-1,\lambda_Nv)$ with a unitary map, which allows us to express the Hamiltonian in terms of excitations around the condensate. This conjugated Hamiltonian can then be expanded in a power series in $\lambda_N^{1/2}$. Then, in Section~\ref{sec_lemma_expansions}, we prove the corresponding expansion of $E(N-1,\lambda v)$ in Proposition~\ref{lemma_energy_expansions}. In Section~\ref{sec_explicit_computations}, we compute $E_1^\binding$ in order to compare our result in detail with \cite{nam2018}. Finally, in Section~\ref{sec_homogeneous_comp}, we compute $E_1^\binding$ and $E_2^\binding$ explicitly for the homogeneous Bose gas, i.e., we prove Theorem~\ref{theorem_homogeneous}.

\section{Proof of the Expansion}\label{sec_main_theorem_proof}

\subsection{The Hamiltonians on the Excitation Fock Space}\label{sec_Ham_Exc_Space}
We fix $\varphi$ to be the solution to the Hartree equation
\begin{equation}
\Big(T + v*|\varphi|^2 - \scp{\varphi}{(T + v*|\varphi|^2)\varphi}\Big) \varphi = 0,
\end{equation}
i.e., $\varphi$ is the minimizer of the Hartree functional \eqref{Hartree_functional}. Let us define
\begin{equation}
h(w) := T + w*|\varphi|^2 - \mu(w), ~\text{with}~ \mu(w) := \scp{\varphi}{(T + w*|\varphi|^2)\varphi},
\end{equation}
and $\eH(w) := \scp{\varphi}{\big(T+\frac{1}{2}w*|\varphi|^2\big) \varphi}$. With this notation $\varphi$ is the solution of $h(v) \varphi = 0$. The $N$-body Hamiltonian \eqref{H_N_w} with interaction $w$ can be rewritten as 
\begin{align}
H(N,w) = N \eH\big((N-1)w\big) + \sum_{j=1}^N h_j\big((N-1)w\big) + \frac{1}{N-1}\sum_{1\leq i<j \leq N} W_{ij}\big((N-1)w\big),
\end{align}
where we defined
\begin{equation}\label{W_def}
W_{ij}(w) := W(w)(x_i,x_j) := w(x_i-x_j) - \big( w * |\varphi|^2\big)(x_i) - \big( w * |\varphi|^2\big)(x_j) + \scp{\varphi}{w * |\varphi|^2 \varphi}.
\end{equation}
With these definitions, the $N$-body Hamiltonian with interaction $w = \lambda_Nv = (N-1)^{-1}v$ is
\begin{equation}\label{Ham_N_body}
H(N) := H(N,\lambda_Nv) = N \eH(v) + \sum_{j=1}^N h_j(v) + \frac{1}{N-1} \sum_{1\leq i<j \leq N} W_{ij}(v),
\end{equation}
and the $(N-1)$-body Hamiltonian with the same coupling constant $\lambda_N$ is
\begin{align}\label{Ham_N_minus_1_body}
\widetilde{H}(N-1) 
&:= H(N-1,\lambda_Nv) \nonumber\\
&= (N-1) \eH(v-\lambda_Nv) + \sum_{j=1}^{N-1} h_j(v-\lambda_Nv) + \frac{1}{N-2} \sum_{1\leq i<j \leq N-1} W_{ij}(v-\lambda_Nv),
\end{align}
where we used $(N-2)\lambda_N v = v - \lambda_N v$. In order to prove Theorem~\ref{theorem_main} we derive asymptotic expansions for the ground state energies of $H(N)$ and $\widetilde{H}(N-1)$ separately and then use the definition \eqref{def_binding_energy} of the binding energy.\footnote{The advantage of this method is that the leading order in the expansions of the ground states of $H(N)$ and $\widetilde{H}(N-1)$ is the same, and, up to a known unitary transformation, independent of $N$ (it is given by Bogoliubov theory, as explained around Equations~\eqref{H_N_on_exc_space_series} and \eqref{H_N_minus_1_on_exc_space_series}). This allows for a simple computation of all the following orders. Alternatively, one may consider treating $\widetilde{H}(N-1)$ as a perturbation of $H(N)$. However, in this approach, the ground state of the unperturbed system will depend on $N$, making computations of the following higher orders more difficult.} The expansion for $H(N)$ was already proven in \cite{spectrum}. The adaption to $\widetilde{H}(N-1)$ requires some modifications since $\widetilde{H}(N-1)$ is not equal to $H(N-1)$ due to the fact that we keep the same coupling constant $\lambda_N$ for both the Hamiltonians $H(N)$ and $\widetilde{H}(N-1)$. In the rest of this section we explain the necessary modifications. With these modifications, we then prove in Section~\ref{sec_lemma_expansions} the expansion of the ground state energy of $\widetilde{H}(N-1)$.

For $f,g \in L^2(\Omega)$, we introduce the usual creation and annihilation operators $a^*(f)$ and $a(f)$, which satisfy the CCR $[a(f),a(g)] = 0 = [a^*(f),a^*(g)]$, $[a(f),a^*(g)] = \scp{f}{g}$. For ease of notation we will often use the operator-valued distributions $a^*_x$ and $a_x$. Denoting by $\overline{f(x)}$ the complex conjugate of $f(x)$, these are defined by
\begin{equation}
a^*(f) = \int \dx f(x) a^*_x, ~~~ a(f) = \int \dx \overline{f(x)} a_x.
\end{equation}
They satisfy the CCR $[a_x,a_y] = 0 = [a^*_x,a^*_y]$ and $[a_x,a^*_y] = \delta(x-y)$. We define the second quantization of a one-body operator $A$ on $L^2(\Omega)$ with integral kernel $A(x,y)$ as
\begin{equation}
\dGamma(A) = \int \dx \dy\, \ad_x A(x,y) a_y.
\end{equation}
In particular, the excitation number operator is given by
\begin{equation}
\Number_\perp := \dGamma(q),
\end{equation}
where $q := 1 - p$ with $p := \ket{\varphi}\!\bra{\varphi}$.

Next, we perform a version of Bogoliubov's $c$-number substitution \cite{bogoliubov1947} as it was introduced in \cite{lewin2015_2}. For this, we define a unitary map
\begin{equation}
\UNp: \mathcal{H}^N_{\sym} \to \Fock_\perp^{\leq N} := \bigoplus_{k=0}^N \bigotimes_\sym^k \{ \varphi \}^\perp, \Psi \mapsto \sum_{j=0}^N q^{\oplus j}\left( \frac{a(\varphi)^{N-j}}{\sqrt{(N-j)!}} \Psi \right).
\end{equation}
We call $\UNp$ the excitation map and $\Fock_\perp^{\leq N}$ the truncated excitation Fock space. Furthermore $\Fock_\perp := \bigoplus_{k=0}^\infty \bigotimes_\sym^k \{ \varphi \}^\perp$ denotes the excitation Fock space without truncation. Note that every wave function $\Psi$ can be decomposed as 
\begin{equation}
\Psi = \sum_{k=0}^N \varphi^{\otimes (N-k)} \otimes_s \chi^{(k)}, ~~ \text{with}~~ \chi^{(k)} \in \bigotimes_\sym^k \{ \varphi \}^\perp,
\end{equation}
and that $\UNp \Psi = \big( \chi^{(0)}, \chi^{(1)}, \ldots, \chi^{(N)} \big)$. For general interactions $w$, we find by an explicit computation, similar as in \cite{lewin2015_2,spectrum} that
\begin{align}\label{Ham_exc_general}
\UNp \, H(N,w) \, \UNp^*
&= N \eH\big((N-1)w\big) + \FockH^{\exc}(N,w) + \FockH^{\extra}(N,w),
\end{align}
with
\begin{align}
\begin{split}
\FockH^{\exc}(N,w)
&= \boldKz\big((N-1)w\big) + \frac{N-\Number_\perp}{N-1}\boldKo\big((N-1)w\big) \\
&\quad + \left( \boldKt\big((N-1)w\big) \frac{\sqrt{(N-\Number_\perp)(N-\Number_\perp-1)}}{N-1} + \hc \right) \\
&\quad + \left( \boldKth\big((N-1)w\big) \frac{\sqrt{N-\Number_\perp}}{N-1} + \hc \right) +  \frac{1}{N-1} \boldKf\big((N-1)w\big),
\end{split}
\end{align}
where $\hc$ denotes the Hermitian conjugate of the preceding term, and
\begin{align}
\FockH^{\extra}(N,w)
&= \sqrt{N-\Number_\perp} a\big( q h\big((N-1)w\big) \varphi \big) + \hc.
\end{align}
Here, we have defined
\begin{subequations}
\begin{align}
\boldKz(w) &:= d\Gamma(qh(w)q), \\[5pt]  
\boldKo(w) &:= d\Gamma(K_1(w)), \\
\boldKt(w) &:= \frac{1}{2} \int \dx_1\dx_2\, K_2(w)(x_1;x_2) a_{x_1}^* a_{x_2}^*, \\
\boldKth(w) &:= \int \dx_1\dx_2\dx_3\, K_3(w)(x_1,x_2;x_3) a_{x_1}^* a_{x_2}^* a_{x_3}, \\
\boldKf(w) &:= \frac{1}{2} \int \dx_1\dx_2\dx_3\dx_4\, K_4(w)(x_1,x_2;x_3,x_4) a_{x_1}^* a_{x_2}^* a_{x_3} a_{x_4},
\end{align}
\end{subequations}
with, setting $K(w)(x,y) := \overline{\varphi(y)} w(x-y) \varphi(x)$,
\begin{subequations}
\begin{align}
K_1(w)(x_1;x_2) &:= \int \dy_1\dy_2 \, q(x_1,y_1) K(w)(y_1,y_2) q(y_2,x_2), \\
K_2(w)(x_1;x_2) &:= \int \dy_1\dy_2 \, q(x_1,y_1) q(x_2,y_2) K(w)(y_1,y_2), \\
K_3(w)(x_1,x_2;x_3) &:= \int \dy_1\dy_2 \, q(x_1,y_1) q(x_2,y_2) W(w)(y_1,y_2) \varphi(y_1) q(y_2,x_3), \\
K_4(w)(x_1,x_2;x_3,x_4) &:= \int \dy_1\dy_2 \, q(x_1,y_1) q(x_2,y_2) W(w)(y_1,y_2) q(y_1,x_3) q(y_2,x_4),
\end{align}
\end{subequations}
where $q(x,y)$ is the integral kernel of $q$ and $W$ was defined in \eqref{W_def}. 

We now map the Hamiltonians to their respective excitations spaces. For the $N$-body Hamiltonian $H(N)$ from \eqref{Ham_N_body}, Equation \eqref{Ham_exc_general} gives
\begin{align}\label{H_N_on_exc_space}
&\UNp \, H(N,\lambda_N v) \, \UNp^* \nonumber\\
&\quad = N \eH\big(v\big) + \FockH^{\exc}(N,\lambda_Nv) + \FockH^{\extra}(N,\lambda_Nv) \nonumber\\
&\quad = N \eH\big(v\big) + \boldKz(v) + \frac{N-\Number_\perp}{N-1}\boldKo(v) + \left( \boldKt(v) \frac{\sqrt{(N-\Number_\perp)(N-\Number_\perp-1)}}{N-1} + \hc \right) \nonumber\\
&\quad\quad + \left( \boldKth(v) \frac{\sqrt{N-\Number_\perp}}{N-1} + \hc \right) + \frac{1}{N-1} \boldKf(v) \nonumber\\
&\quad =: N \eH\big(v\big) + \FockH(N),
\end{align}
since $\FockH^{\extra}(N,\lambda_Nv) = 0$ due to $h(v)\varphi = 0$. For the $(N-1)$-body Hamiltonian $\widetilde{H}(N-1)$ from \eqref{Ham_N_minus_1_body}, Equation \eqref{Ham_exc_general} yields
\begin{align}\label{H_N_minus_1_on_exc_space}
&\UNmop \, H(N-1,\lambda_N v) \, \UNmop^* \nonumber\\
&\quad = (N-1) \eH\big((N-2)\lambda_Nv\big) + \FockH^{\exc}(N-1,\lambda_Nv) + \FockH^{\extra}(N-1,\lambda_Nv) \nonumber\\
&\quad = (N-1) \eH(v) - \frac{1}{2} \scp{\varphi}{(v*|\varphi|^2)\varphi} + \boldKz(v) - \lambda_N \dGamma\big( q \big[ v*|\varphi|^2 - \scp{\varphi}{v*|\varphi|^2\varphi} \big] q \big) \nonumber\\
&\quad\quad + \frac{N-\Number_\perp-1}{N-1}\boldKo(v) + \left( \boldKt(v) \frac{\sqrt{(N-1-\Number_\perp)(N-2-\Number_\perp)}}{N-1} + \hc \right) \nonumber\\
&\quad\quad + \left( \Big( \boldKth(v) - a^*\big( q \big( v*|\varphi|^2 \big) \varphi \Big) \frac{\sqrt{N-1-\Number_\perp}}{N-1} + \hc \right) +  \lambda_N \boldKf(v) \nonumber\\
&\quad =: (N-1) \eH(v) - \frac{1}{2} \scp{\varphi}{(v*|\varphi|^2)\varphi} + \widetilde{\FockH}(N-1),
\end{align}
where we used $h(v)\varphi= 0$. Note that here there is a contribution from $\FockH^{\extra}$.

Next, we expand $\FockH(N):\Fock_\perp^{\leq N} \to \Fock_\perp^{\leq N}$ and $\widetilde{\FockH}(N-1):\Fock_\perp^{\leq N-1} \to \Fock_\perp^{\leq N-1}$ in power series in $\lambda_N^{1/2}$. We begin with $\FockH(N)$. Following \cite[Def.~3.9]{spectrum}, it is convenient to extend $\FockH(N)$ to an operator on $\Fock_\perp$ as $\FockH(N) \oplus E_N^{(-1)}$, where $E_N^{(-1)} := E_N^{(0)} - (E_N^{(1)} - E_N^{(0)})$, with $E_N^{(n)}$ the eigenvalues of $\FockH(N)$. Note that $E_N^{(0)}$ is non-degenerate, so $E_N^{(-1)} < E_N^{(0)}$. We continue to denote this extended operator by $\FockH(N)$. Following \cite[Section~3.2]{spectrum}, it is furthermore convenient to treat the particle number conserving terms in $\FockH(N)$ acting on $\Fock_\perp^{>N} := \bigoplus_{k=N+1}^\infty \bigotimes_\sym^k \{ \varphi \}^\perp$ separately. Thus, we write
\begin{align}
\FockH(N) = \FockH^<(N) + \FockH^>(N),
\end{align}
with
\begin{align}
\begin{split}
\FockH^<(N) &:= \boldKz(v) + \frac{N-\Number_\perp}{N-1}\boldKo(v) + \left( \boldKt(v) \frac{\sqrt{[(N-\Number_\perp)(N-\Number_\perp-1)]_+}}{N-1} + \hc \right) \\
&\quad + \left( \boldKth(v) \frac{\sqrt{[N-\Number_\perp]_+}}{N-1} + \hc \right) + \frac{1}{N-1} \boldKf(v),
\end{split}
\end{align}
where $[\cdot]_+$ denotes the positive part, and
\begin{align}
\FockH^>(N) &:= 0 \oplus \left( E_N^{(-1)} - \boldKz(v) - \frac{N-\Number_\perp}{N-1}\boldKo(v) - \frac{1}{N-1} \boldKf(v) \right),
\end{align}
where $\oplus$ is to be understood w.r.t.\ the decomposition $\Fock_\perp = \Fock_\perp^{\leq N} \oplus \Fock_\perp^{> N}$. Here, we added in $\FockH^<(N)$ the action of the particle number conserving terms on $\Fock^{> N}$ and subtracted them again in $\FockH^>(N)$. Then, a Taylor expansion of the square roots allows us to write, for any  $a \in \N$,
\begin{align}\label{H_N_on_exc_space_series}
\FockH^<(N) = \FockH_0 + \sum_{j=1}^a \lambda_N^{j/2} \FockH_j + \lambda_N^{(a+1)/2} \FockR_a
\end{align}
as was shown in \cite[Proposition~3.12]{spectrum}. Here, $\FockH_0$ is the Bogoliubov Hamiltonian
\begin{align}\label{Bog_Ham}
\FockH_0 = \boldKz(v) + \boldKo(v) + \Big( \boldKt(v) + \hc \Big),
\end{align}
and
\begin{subequations}
\begin{align}
\FockH_1 &:= \boldKth(v) + \hc, \label{H_1_def}\\
\FockH_2 &:= - (\Number_\perp-1)\boldKo(v) - \left( \boldKt(v) \left(\Number_\perp-\frac{1}{2}\right) +\hc \right) + \boldKf(v), \\
\FockH_{2j-1} &:= c_{j-1} \left( \boldKth(v) \big( \Number_\perp-1 \big)^{j-1} + \hc \right), \\
\FockH_{2j} &:= \sum_{\nu=0}^j d_{j,\nu} \Big( \boldKt(v) \big( \Number_\perp-1 \big)^{\nu} + \hc \Big),
\end{align}
\end{subequations}
for $j \geq 2$, with coefficients
\begin{subequations}
\begin{align}
c_0^{(\ell)} &:= 1, ~~ c_j^{(\ell)} := \frac{(\ell-\frac{1}{2})(\ell+\frac{1}{2})(\ell+\frac{3}{2}) \cdots (\ell+ j-\frac{3}{2})}{j!}, ~~ c_j := c_j^{(0)} ~ (j \geq 1), \\
d_{j,\nu} &:= \sum_{\ell=0}^\nu c_\ell^{(0)} c_{\nu-\ell}^{(0)} c_{j-\nu}^{(\ell)} ~~(j \geq \nu \geq 0).
\end{align}
\end{subequations}
The remainder $\FockR_a$, defined by \eqref{H_N_on_exc_space_series}, still depends on $N$, but can be estimated uniformly in $N$ in terms of powers of the number operator, and in terms of $\Number_\perp$ and $\FockH_0$ for $a\leq 2$ if $v$ is unbounded; see \cite[Lemma~3.11 and Lemma~5.2]{spectrum}. 

We now turn to $\widetilde{\FockH}(N-1)$. Analogously to above, we extend it to an operator on $\Fock_\perp$, and write it as
\begin{align}
\widetilde{\FockH}(N-1) = \widetilde{\FockH}^<(N-1) + \widetilde{\FockH}^>(N-1),
\end{align}
with
\begin{align}
\begin{split}
\widetilde{\FockH}^<(N-1) &:= \boldKz(v) - \lambda_N \dGamma\big( q \big[ v*|\varphi|^2 - \scp{\varphi}{v*|\varphi|^2\varphi} \big] q \big) \\
&\quad + \frac{N-\Number_\perp-1}{N-1}\boldKo(v) + \left( \boldKt(v) \frac{\sqrt{[(N-1-\Number_\perp)(N-2-\Number_\perp)]_+}}{N-1} + \hc \right) \\
&\quad + \left( \Big( \boldKth(v) - a^*\big( q \big( v*|\varphi|^2 \big) \varphi \Big) \frac{\sqrt{[N-1-\Number_\perp]_+}}{N-1} + \hc \right) +  \lambda_N \boldKf(v),
\end{split}
\end{align}
and
\begin{align}\label{def_of_Fock_tilde_bigger}
\begin{split}
\widetilde{\FockH}^>(N-1) &:= 0 \oplus \Big( \widetilde{E}_{N-1}^{(-1)} - \boldKz(v) + \lambda_N \dGamma\big( q \big[ v*|\varphi|^2 - \scp{\varphi}{v*|\varphi|^2\varphi} \big] q \big) \\
&\qquad\qquad - \frac{N-\Number_\perp-1}{N-1}\boldKo(v) - \frac{1}{N-1} \boldKf(v) \Big),
\end{split}
\end{align}
where here $\oplus$ is to be understood w.r.t.\ the decomposition $\Fock_\perp = \Fock_\perp^{\leq N-1} \oplus \Fock_\perp^{> N-1}$, and $\widetilde{E}_{N-1}^{(n)}$ denote the eigenvalues of $\widetilde{\FockH}(N-1)$, with $\widetilde{E}_{N-1}^{(-1)} := \widetilde{E}_{N-1}^{(0)} - (\widetilde{E}_{N-1}^{(1)} - \widetilde{E}_{N-1}^{(0)})$. We then expand $\widetilde{\FockH}^<(N-1)$ for any  $a \in \N$ as
\begin{align}\label{H_N_minus_1_on_exc_space_series}
\widetilde{\FockH}^<(N-1) = \FockH_0 + \sum_{j=1}^a \lambda_N^{j/2} \widetilde{\FockH}_j + \lambda_N^{(a+1)/2} \widetilde{\FockR}_a,
\end{align}
where
\begin{subequations}\label{tilde_Hamiltonians}
\begin{align}
\widetilde{\FockH}_1 &:= \Big( \boldKth(v) - a^*\big( q \big( v*|\varphi|^2 \big) \varphi \big) \Big) + \hc, \\
\begin{split}
\widetilde{\FockH}_2 &:= - \dGamma\big( q \big[ v*|\varphi|^2 - \scp{\varphi}{v*|\varphi|^2\varphi} \big] q \big) \\
&\quad - \Number_\perp\boldKo(v) - \left( \boldKt(v) \left(\Number_\perp+\frac{1}{2}\right) +\hc \right) + \boldKf(v),
\end{split} 
\\
\widetilde{\FockH}_{2j-1} &:= c_{j-1} \Big( \boldKth(v) - a^*\big( q \big( v*|\varphi|^2 \big) \varphi \big) \Big) \Number_\perp^{j-1} + \hc, \\
\widetilde{\FockH}_{2j} &:= \sum_{\nu=0}^j d_{j,\nu} \Big( \boldKt(v) \Number_\perp^{\nu} + \hc \Big)
\end{align}
\end{subequations}
for $j \geq 2$. The remainder $\widetilde{\FockR}_a$ can be bounded analogously to $\FockR_a$, in particular uniformly in $N$, as we will explain in the proof of Proposition~\ref{lemma_energy_expansions}. Note that the leading order term $\FockH_0$ is the same in the expansions \eqref{H_N_on_exc_space_series} and \eqref{H_N_minus_1_on_exc_space_series}. The $\widetilde{\FockH}_j$ Hamiltonians differ from the $\FockH_j$ in the following way:
\begin{itemize}
\item $\boldKth(v)$ is replaced by $\boldKth(v) - a^*\big( q \big( v*|\varphi|^2 \big) \varphi \big) =: \widetilde{\boldK}_3(v)$,
\item an extra term $- \dGamma\big( q \big[ v*|\varphi|^2 - \scp{\varphi}{v*|\varphi|^2\varphi} \big] q \big)$ is added for $j=2$,
\item $\Number_\perp$ is replaced by $\Number_\perp+1$.
\end{itemize}
Note that the formulas \eqref{tilde_Hamiltonians} can be simplified by using the properties of the coefficients $c_j$ and $d_{j,\nu}$. Equivalently we could use the fact that \eqref{tilde_Hamiltonians} can be obtained from replacing $N \to N-1$, $v \to \frac{N-2}{N-1}v$, and $\boldKth \to \widetilde{\boldK}_3$ in the Taylor expansion of $\FockH^<(N)$ from \eqref{H_N_on_exc_space_series} in all terms except the constant terms and those involving $\boldKz$. Then the $\widetilde{\FockH}_j$ for $j=1$ and $j \geq 3$ can be expressed in terms of the $\FockH_j$ if one additionally replaces $\boldKth$ by $\widetilde{\boldK}_3$ wherever it occurs. For example, we find for $j=1,2,3,4$ that
\begin{subequations}\label{H_tilde_1_4}
\begin{align}
\widetilde{\FockH}_1 &= \FockH_1 \big|_{\boldKth \to \widetilde{\boldK}_3} \label{def_H_tilde_1_and_2}, &&\widetilde{\FockH}_2
= \FockH_2 - \FockH_0 + \dGamma\big( q \big[ T - \scp{\varphi}{T\varphi} \big] q \big), \\
\widetilde{\FockH}_3 &= \Big( \FockH_3 - \frac{1}{2} \FockH_1 \Big) \Big|_{\boldKth \to \widetilde{\boldK}_3}, &&\widetilde{\FockH}_4 = \FockH_4. \label{def_H_tilde_3_and_4}
\end{align}
\end{subequations}

\subsection{Expansions of the Ground State Energies}\label{sec_lemma_expansions}
One of the main results of \cite{spectrum} is an expansion of the ground state energy of $H(N,\lambda_N v)$ in powers of $\lambda_N$. Using the computations from Section~\ref{sec_Ham_Exc_Space} we can adapt this result to yield an expansion of the ground state energy of $H(N-1,\lambda_N v)$ in powers of $\lambda_N$ as well. We denote the unique ground state of $\FockH_0$ from Equation~\eqref{Bog_Ham} by $\Chi_0$, its ground state energy by $E_0$, introduce the projections
\begin{equation}
\FockP_0 := | \Chi_0 \rangle \langle \Chi_0 |, ~~\FockQ_0 := 1- \FockP_0,
\end{equation}
and define
\begin{equation}\label{FockO}
\FockO_k:=\begin{cases}
\displaystyle\quad-\Pz &\quad k=0\,,\\[7pt]
\displaystyle\frac{\Qz}{\big(E_0-\FockHz\big)^k} &\quad k>0\,.
\end{cases}
\end{equation}
Then the following holds.

\begin{proposition}\label{lemma_energy_expansions}
Let $a \in \N_0$ and let Assumptions~\ref{Assumption_1}, \ref{Assumption_2}, and \ref{Assumption_3} hold. Then for sufficiently large $N$ there exist $C(a)>0$ such that
\begin{equation}\label{expansion_E_N}
\left| E(N,\lambda_N v) - N \eH - E_0 - \sum_{\ell=1}^a \lambda_N^\ell E_{\ell} \right| \leq C(a) \lambda_N^{a+1},
\end{equation}
with
\begin{equation}\label{def_E_l}
E_{\ell} = \sum_{\nu = 1}^{2\ell} \sum_{\substack{\bj \in \N^\nu \\ |\bj| = 2\ell}} \sum\limits_{\substack{\bm\in\N_0^{\nu-1}\\|\bm|=\nu-1}}\frac{1}{\kappa(\bm)} \SCP{\Chi_0}{\FockHjo\FockO_{m_1}\mycdots\FockH_{j_{\nu-1}}\FockO_{m_{\nu-1}}\FockHjnu \Chi_0},
\end{equation}
where $\kappa(\bm):=1+\left|\left\{\mu\,:\, m_\mu=0\right\}\right| \in\{1\mydots \nu-1\}$ is the number of operators $\Pz$ within the scalar product. Furthermore,
\begin{equation}\label{expansion_E_N_minus_1}
\left| E(N-1,\lambda_N v) - (N-1) \eH + \frac{1}{2} \scp{\varphi}{(v*|\varphi|^2)\varphi} - E_0 - \sum_{\ell=1}^a \lambda_N^\ell \widetilde{E}_{\ell} \right| \leq C(a) \lambda_N^{a+1},
\end{equation}
where $\widetilde{E}_{\ell}$ is defined by Equation~\eqref{def_E_l} with $\FockH_j$ replaced by $\widetilde{\FockH}_j$. The coefficients from Theorem~\ref{theorem_main} are given by $E^{\mathrm{binding}}_j = E_j - \widetilde{E}_j$.
\end{proposition}

Note that $\Chi_0$ and $\FockO_k$ are the same in the formulas for $E_\ell$ and $\widetilde{E}_{\ell}$, since the leading order of both $\FockH^<(N)$ and $\widetilde{\FockH}^<(N-1)$ is the same, namely $\FockHz$.

\begin{proof}
The estimate \eqref{expansion_E_N} is proven in \cite[Theorem~2]{spectrum}. It is based on Rayleigh--Schr\"odinger perturbation theory applied to $\FockH(N)$. More exactly, a rigorous expansion of the projection $\FockP$ on the ground state of $\FockH(N)$ is proven, and based on that an expansion of the ground state energy $E = \Tr\, \FockH(N) \FockP$. The estimate \eqref{expansion_E_N_minus_1} can be obtained with the same strategy, but here the underlying Hamiltonian $\widetilde{\FockH}(N-1)$ is different. In order to make the proof from \cite{spectrum} work, two things have to be checked:
\begin{enumerate}[label={(\alph*)}]
\item \textbf{Estimates for $\widetilde{\FockH}_j$ and $\widetilde{\FockR}_a$.} First, note that \cite[Lemma~5.2]{spectrum} still holds when we replace $\boldKth$ by $\widetilde{\boldK}_3$ and $\FockH(N)$ by $\widetilde{\FockH}(N-1)$, i.e., we still have
\begin{subequations}
\begin{align}
\| \widetilde{\boldK}_3^{(*)} \phi \| &\leq C \| (\Number_\perp +1)^{3/2} \phi \|, \\
\Big\| \Big[ \widetilde{\FockH}^<(N-1), (\Number_\perp+1)^\ell \Big] \phi \Big\|_{\Fock_\perp^{\leq N-1}} &\leq C(\ell) \| (\Number_\perp+1)^\ell \phi \|_{\Fock_\perp^{\leq N-1}},
\end{align}
\end{subequations}
for some $C > 0$ and $C(\ell) >0$, and for all $\phi \in \Fock_\perp$. Additionally, we have
\begin{align}
\big\| \widetilde{\FockH}_2 \phi \big\| \leq C \| (\Number_\perp + 1)^2 \phi \|
\end{align}
for some $C>0$ and for all $\phi \in \Fock_\perp$, so we can use the same bounds for $\widetilde{\FockH}_2$ as we have used for $\FockH_2$ in \cite{spectrum}. Since \cite[Lemma~5.3~(a)]{spectrum} is proven directly by using \cite[Lemma~5.2]{spectrum}, it also holds when $\FockH_j$ is replaced by $\widetilde{\FockH}_j$. The estimates for $\widetilde{\FockR}_a$ are obtained analogously.
\item \textbf{Occurrence of $\widetilde{\FockH}^>(N-1)$.} In Equation~\eqref{def_of_Fock_tilde_bigger} we have defined $\widetilde{\FockH}^>(N-1)$ in such a way that \cite[Proposition~3.14]{spectrum} can be applied, meaning that the operator $\widetilde{\FockH}^>(N-1)$ does not contribute to $\FockP$.
\end{enumerate}
Thus, the proof of \cite[Theorem~2]{spectrum} still works when we replace $\FockH(N)$ by $\widetilde{\FockH}(N-1)$, meaning that \eqref{expansion_E_N_minus_1} holds.
\end{proof}

\section{Explicit Computations}\label{sec_explicit_computations}

We use the notation
\begin{equation}\label{H_1_extra_def}
\FockH_1^\extra := \FockH_1 - \widetilde{\FockH}_1 = a^*\big( q \big( v*|\varphi|^2 \big) \varphi \big) + \hc,
\end{equation}
and abbreviate $\FockO := \FockO_1$.

\subsection{The Trapped Bose Gas}\label{sec_next_order_Nam}

In this section we compute $E_1^\binding$ for the trapped Bose gas and compare the result with Nam's conjecture \cite[Conjecture~6]{nam2018}. For $\ell = 1$, the formula \eqref{def_E_l} is $E_1 = \scp{\Chi_0}{\FockH_2 \Chi_0} + \scp{\Chi_0}{\FockH_1 \FockO \FockH_1 \Chi_0}$. 
Thus, using the formulas \eqref{def_H_tilde_1_and_2}, we find
\begin{align}\label{E_1_binding}
E_1^\binding &= E_1 - \widetilde{E}_1 \nonumber\\
\begin{split}
&= \scp{\Chi_0}{\FockH_0 \Chi_0} - \scp{\Chi_0}{\dGamma\big( q \big[ T - \scp{\varphi}{T\varphi} \big] q \big) \Chi_0} + 2 \Re \scp{\Chi_0}{\FockH_1^\extra \FockO \FockH_1 \Chi_0} \\
&\quad - \scp{\Chi_0}{\FockH_1^\extra \FockO \FockH_1^\extra \Chi_0}.
\end{split}
\end{align}
The main part of the conjecture is that
\begin{align}\label{def_coeff_C}
E_1^\binding = C := \scp{\Chi_0}{\FockH_0 \Chi_0} - \scp{\Chi_0}{\dGamma\big( q \big[ T - \scp{\varphi}{T\varphi} \big] q \big) \Chi_0},
\end{align}
which does not in general agree with the correct expression \eqref{E_1_binding}. Note, however, that it does agree for the homogeneous Bose gas, since then $\FockH_1^\extra = 0$. 

The discrepancy can be explained as follows. Let $\Chi$ denote the ground state of $\FockH(N)$. Our results imply that
\begin{equation}
B_N := 2 \Re \scp{\Chi}{a^*\big( q \big( v*|\varphi|^2 \big) \varphi \big) \Chi} = N^{-1/2} 2 \Re \scp{\Chi_0}{\FockH_1^\extra \FockO \FockH_1 \Chi_0} + O(N^{-3/2}),
\end{equation}
i.e, this term is $O(N^{-1/2})$. This is in contrast to the prediction $B_N = o(N^{-1/2})$ from \cite{nam2018}. Moreover, a closer look at the estimates in \cite{nam2018} reveals that two bounds are proven, namely
\begin{subequations}
\begin{align}
E_1^\binding &\geq C + N^{1/2} B_N + o(1), \\
E_1^\binding &\leq C + N^{1/2} B_N + 2D + o(1),
\end{align}
\end{subequations}
where $D := - \scp{\Chi_0}{\FockH_1^\extra \FockO \FockH_1^\extra \Chi_0} \geq 0$. The correct expression in the limit $N\to\infty$, however, is as in \eqref{E_1_binding}, i.e., $E_1^\binding = C + N^{1/2} B_N + D$.

\subsection{The Homogeneous Bose Gas}\label{sec_homogeneous_comp}
For the homogeneous Bose gas $\varphi(x) = 1$, which implies $v*|\varphi|^2 = \hat{v}(0)$, $q \big( v*|\varphi|^2 \big) \varphi = 0$ and thus $\FockH_1^\extra = 0$, and $T\varphi= 0$. Then the formulas \eqref{H_tilde_1_4} simplify to
\begin{align}\label{H_tilde_1_4_torus}
\tildeFockH_1 = \FockH_1,\qquad \tildeFockH_2 = \FockH_2 - \FockH_0 + \dGamma(qTq),\qquad \tildeFockH_3 = \FockH_3 - \frac{1}{2} \FockH_1,\qquad \tildeFockH_4 = \FockH_4.
\end{align}
Thus Equation~\eqref{E_1_binding} becomes
\begin{align}\label{E_1_binding_torus}
E_1^\binding = \scp{\Chi_0}{\big( \FockH_0 - \dGamma(qTq) \big) \Chi_0} = E_0 - \scp{\Chi_0}{\dGamma(qTq)\Chi_0}.
\end{align}
In order to compute $E_2^\binding$, note that Equation~\eqref{def_E_l} for $\ell = 2$ can be written as
\begin{align}
\begin{split}
E_2 &= \scp{\Chi_0}{\FockH_4 \Chi_0} + \scp{\Chi_0}{\FockH_3 \FockO \FockH_1 \Chi_0} + \scp{\Chi_0}{\FockH_1 \FockO \FockH_3 \Chi_0} + \scp{\Chi_0}{\FockH_2 \FockO \FockH_2 \Chi_0} \\
&\quad + \scp{\Chi_0}{\FockH_2 \FockO \FockH_1 \FockO \FockH_1 \Chi_0} + \scp{\Chi_0}{\FockH_1 \FockO (\FockH_2 - E_1) \FockO \FockH_1 \Chi_0} + \scp{\Chi_0}{\FockH_1 \FockO \FockH_1 \FockO \FockH_2 \Chi_0} \\
&\quad + \scp{\Chi_0}{\FockH_1 \FockO \FockH_1 \FockO \FockH_1 \FockO \FockH_1 \Chi_0}.
\end{split}
\end{align}
Then a computation using \eqref{H_tilde_1_4_torus}, $\FockH_0 \Chi_0 = E_0 \Chi_0$, $\FockO \Chi_0 = 0$, and $\FockH_0 \FockO = \FockH_0 \frac{\FockQ_0}{E_0 - \FockH_0} = - \FockQ_0 + E_0 \FockO$ yields
\begin{align}\label{E_2_binding}
E_2^\binding &:= E_2 - \widetilde{E}_2 \nonumber\\
\begin{split}
&= - 2 \Re \scp{\Chi_0}{\dGamma(qTq)\FockO \FockH_2 \Chi_0} - 2 \Re \scp{\Chi_0}{\dGamma(qTq) \FockO \FockH_1 \FockO \FockH_1 \Chi_0} \\
&\quad - \scp{\Chi_0}{\dGamma(qTq) \FockO \dGamma(qTq) \Chi_0} - \scp{\Chi_0}{\FockH_1 \FockO \Big( \dGamma(qTq) - \scp{\Chi_0}{\dGamma(qTq)\Chi_0} \Big) \FockO \FockH_1 \Chi_0}.
\end{split}
\end{align}
In the rest of this section all summations are over the lattice $(2\pi \Z)^d$. In Fourier representation, the operators $\FockH_0$, $\dGamma(qTq)$, $\FockH_1$, and $\FockH_2$ read
\begin{subequations}
\begin{align}
\FockH_0 &= \sum_{k \neq 0} \big( k^2 +  \hat{v}(k) \big) \ad_k a_k + \frac{1}{2} \sum_{k \neq 0} \hat{v}(k) \big( \ad_k \ad_{-k} + a_k a_{-k} \big), \label{H_0_torus}\\
\dGamma(qTq) &= \sum_{k \neq 0} k^2 \ad_k a_k, \label{T_torus}\\
\FockH_1 &= \sum_{\substack{k,\ell \neq 0 \\ k + \ell \neq 0}} \hat{v}(k) \ad_k\ad_\ell a_{k+\ell} + \hc, \\
\begin{split}
\FockH_2 &= - \sum_{k,\ell \neq 0} \hat{v}(k) \ad_\ell\ad_k a_\ell a_k -\frac{1}{2} \left( \sum_{k\neq 0} \hat{v}(k)\ad_k\ad_{-k} \bigg( \sum_{\ell \neq 0} \ad_\ell a_\ell - \frac{1}{2} \bigg) + \hc \right) \\
&\quad + \frac{1}{2}\sum_{\substack{j,k,\ell \neq 0 \\ j-\ell \neq 0, j+k-\ell\neq 0}} \hat{v}(j-\ell) \ad_j\ad_k a_\ell a_{j+k-\ell}.
\end{split}
\end{align}
\end{subequations}
Furthermore, the Bogoliubov transformation $U_B$ that diagonalizes the Bogoliubov Hamiltonian $\FockH_0$ acts on creation and annihilation operators in the following way. For $p\neq 0$,
\begin{subequations}
\begin{align}
U_B a_p U_B^* &= \sigma_p a_p - \gamma_p \ad_{-p}, \\
U_B \ad_p U_B^* &= \sigma_p \ad_p - \gamma_p a_{-p},
\end{align}
\end{subequations}
with $\sigma_p$ and $\gamma_p$ defined in \eqref{epsilon_alpha_sigma_gamma}. Then
\begin{equation}
U_B \FockH_0 U_B^* = E_0 + \sum_{k \neq 0} \varepsilon(k) \ad_k a_k, ~~\text{with}~~ \varepsilon(k) = \sqrt{k^4 + 2 k^2 \hat{v}(k)},
\end{equation}
so the ground state of $\FockH_0$ is $\ket{\Chi_0} = U_B^* \ket{\Omega}$. The unitary map $U_B$ consequently diagonalizes $\FockO$ as well and we find
\begin{align}
U_B \FockO U_B^* \, \ad_{p_1} \ldots \ad_{p_n}\ket{\Omega} = - \frac{1}{\varepsilon(p_1) + \ldots + \varepsilon(p_n)} \, \ad_{p_1} \ldots \ad_{p_n} \ket{\Omega}
\end{align}
for all $0 \neq p_1,\ldots,p_n \in (2\pi \Z)^d$. We can now compute $E_1^\binding$ and $E_2^\binding$ explicitly.

\vspace{2mm}
\noindent\textbf{Computation of $E_1^\binding$.} Using \eqref{H_0_torus} and \eqref{T_torus} in \eqref{E_1_binding_torus} we find
\begin{align}\label{E_1_binding_torus_computed}
E_1^\binding &= \sum_{k \neq 0} \scp{\Omega}{U_B \Big( \hat{v}(k) \ad_k a_k + \frac{1}{2} \hat{v}(k) \big( \ad_k \ad_{-k} + a_k a_{-k} \big) \Big) U_B^* \Omega} 
= - \sum_{k \neq 0} \hat{v}(k) \frac{\alpha_k}{1+\alpha_k}.
\end{align}
We can now either use a direct computation based on the definition of $\alpha_p$ to show that \eqref{E_1_binding_theorem} holds, or we directly compute with \eqref{T_torus} that
\begin{align}\label{E_1_binding_torus_computed_alt}
E_1^\binding &= E_0 - \sum_{k \neq 0} k^2 \scp{\Omega}{U_B \ad_k a_k U_B^* \Omega} = E_0 - \sum_{k \neq 0} k^2 \gamma_k^2 = E_0 - \sum_{k \neq 0} \frac{k^2 \alpha_k^2}{1-\alpha_k^2}.
\end{align}

\vspace{2mm}
\noindent\textbf{Computation of $E_2^\binding$.} We compute each term in \eqref{E_2_binding} separately. First, note that
\begin{align}
U_B \dGamma(qTq) U_B^* = \sum_{k \neq 0} k^2 \Big( \big(\sigma_k^2+\gamma_k^2\big) \ad_k a_k - \sigma_k\gamma_k \ad_k \ad_{-k} - \sigma_k\gamma_k a_{-k} a_k  + \gamma_k^2 \Big).
\end{align}
Then, in order to compute 
\begin{align}
\scp{\Chi_0}{\dGamma(qTq)\FockO \FockH_2 \Chi_0} =\sum_{\ell \neq 0} \ell^2 \sigma_\ell (-\gamma_\ell)\, \scp{\ad_\ell\ad_{-\ell} \Omega}{(U_B\FockO U_B^*) U_B\FockH_2 U_B^* \Omega},
\end{align}
we only need to know the part of $U_B \FockH_2 U_B^*$ with two $\ad$ operators, since $U_B\FockO U_B^*$ is particle-number conserving. We find 
\begin{align}
\scp{\ad_k \ad_{-k} \Omega}{U_B \FockH_2 U_B^* \Omega} = 2 H^{\QP}_{2,\ad\ad}(k)
\end{align}
with
\begin{align}
\begin{split}
H^{\QP}_{2,\ad\ad}(k) &= - \frac{1}{2} \sum_{\substack{\ell \neq 0 \\ \ell \neq k}} \hat{v}(k-\ell) \gamma_\ell \Big( \sigma_k^2 \sigma_\ell + 2 \sigma_k \gamma_\ell \gamma_k + \sigma_\ell \gamma_k^2 \Big) - \frac{1}{2} \hat{v}(k) (\sigma_k - \gamma_k)^2 \sum_{\ell \neq 0} \gamma_\ell^2 \\
&\quad - \sigma_k \gamma_k \sum_{\ell \neq 0} \hat{v}(\ell) \gamma_\ell (\sigma_\ell - \gamma_\ell) + \hat{v}(k) \gamma_k (\sigma_k - \gamma_k)^3 + \frac{1}{4} \hat{v}(k) \big( \sigma_k^2 + \gamma_k^2 \big).
\end{split}
\end{align}
This yields
\begin{align}\label{hom_T_H_2_term}
- 2 \Re \scp{\Chi_0}{\dGamma(qTq)\FockO \FockH_2 \Chi_0} = -2 \sum_{k \neq 0} k^2 \gamma_k \sigma_k \frac{H^{\QP}_{2,\ad\ad}(k)}{\varepsilon(k)}.
\end{align}
Next, we compute directly that
\begin{align}\label{hom_T_T_term}
- \scp{\Chi_0}{\dGamma(qTq) \FockO \dGamma(qTq) \Chi_0} = \sum_{k \neq 0} \frac{k^4 \sigma_k^2 \gamma_k^2}{\varepsilon(k)}.
\end{align}
For the computation of the remaining terms, note first that
\begin{align}
U_B \FockH_1 U_B^* = \sum_{\substack{k,\ell \neq 0 \\ k + \ell \neq 0}} \left( H^\QP_{1,\ad\ad\ad}(k,\ell) \ad_k\ad_\ell \ad_{-k-\ell} + H^\QP_{1,\ad a a}(k,\ell) \ad_{k+\ell} a_k a_\ell \right) +\hc,
\end{align}
where $H^\QP_{1,\ad\ad\ad}(k,\ell)$ and $H^\QP_{1,\ad a a}(k,\ell)$ can be written in symmetrical form as
\begin{subequations}
\begin{align}
H^\QP_{1,\ad\ad\ad}(k,\ell) &= - \frac{1}{6} \Bigg[ \hat{v}(k) \big(\gamma_{k+\ell} \sigma_\ell + \sigma_{k+\ell} \gamma_\ell\big)\big(\sigma_k - \gamma_k\big) + \hat{v}(\ell) \big(\gamma_{k+\ell} \sigma_k + \sigma_{k+\ell} \gamma_k\big)\big(\sigma_\ell - \gamma_\ell\big) \nonumber\\
&\qquad + \hat{v}(k+\ell) \big(\sigma_\ell\gamma_k + \sigma_k\gamma_\ell\big)\big(\sigma_{k+\ell} - \gamma_{k+\ell}\big) \Bigg],
\\
H^\QP_{1,\ad a a}(k,\ell) &= \frac{1}{2} \Bigg[ \hat{v}(k) \big(\sigma_{k+\ell} \sigma_\ell + \gamma_{k+\ell} \gamma_\ell\big)\big(\sigma_k - \gamma_k\big) + \hat{v}(\ell) \big(\sigma_{k+\ell} \sigma_k + \gamma_{k+\ell} \gamma_k\big)\big(\sigma_\ell - \gamma_\ell\big) \nonumber\\
&\qquad - \hat{v}(k+\ell) \big(\sigma_\ell\gamma_k + \sigma_k\gamma_\ell\big)\big(\sigma_{k+\ell} - \gamma_{k+\ell}\big) \Bigg].
\end{align}
\end{subequations}
With that we find
\begin{align}\label{hom_T_H_1_H_1_term}
&- 2 \Re \scp{\Chi_0}{\dGamma(qTq) \FockO \FockH_1 \FockO \FockH_1 \Chi_0} \nonumber\\
&\quad = 12  \sum_{\substack{k,\ell \neq 0 \\ k + \ell \neq 0}} (k+\ell)^2 \sigma_{k+\ell} \gamma_{k+\ell} \left(\frac{H^\QP_{1,\ad a a}(k,\ell)}{\varepsilon(k+\ell)}\right) \left(\frac{H^\QP_{1,\ad\ad \ad}(k,\ell)}{\varepsilon(k) + \varepsilon(\ell) + \varepsilon(k+\ell)}\right).
\end{align}
Furthermore,
\begin{align}\label{hom_H_1_T_H_1_term}
&- \scp{\Chi_0}{\FockH_1 \FockO \Big( \dGamma(qTq) - \scp{\Chi_0}{\dGamma(qTq)\Chi_0} \Big) \FockO \FockH_1 \Chi_0} \nonumber\\
&\quad = - 18 \sum_{\substack{k,\ell \neq 0 \\ k + \ell \neq 0}} (k+\ell)^2 \big(\sigma_{k+\ell}^2 + \gamma_{k+\ell}^2\big) \Bigg( \frac{H^\QP_{1,\ad\ad\ad}(k,\ell)}{\varepsilon(k) + \varepsilon(\ell) + \varepsilon(k+\ell)} \Bigg)^2.
\end{align}
Adding up \eqref{hom_T_H_2_term}, \eqref{hom_T_T_term}, \eqref{hom_T_H_1_H_1_term}, and \eqref{hom_H_1_T_H_1_term} yields the expression \eqref{E_2_binding_theorem} from Theorem~\ref{theorem_homogeneous}.

\section*{Acknowledgments}
It is a pleasure to thank Phan Th\`anh Nam for helpful discussions on bosonic atoms. L.B.\ was supported by the  German Research Foundation within the Munich Center of Quantum Science and Technology (EXC 2111). N.L.\ gratefully acknowledges support from the Swiss National Science Foundation through the NCCR SwissMap and funding from the European Union's Horizon 2020 research and innovation programme under the Marie Sk\l odowska-Curie grant agreement N\textsuperscript{o} 101024712. S.P.\ acknowledges funding by the Deutsche Forschungsgemeinschaft (DFG, German Research Foundation) – project number 512258249.

\vspace{3mm}

\noindent \textbf{Data Availability.} Data sharing not applicable to this article as no datasets were generated or analyzed during the current study.

\vspace{3mm}

\noindent \textbf{Conflict of interest.} The authors have no relevant financial or non-financial interests to disclose.

\bibliographystyle{abbrv}
    \bibliography{bib_file}

\end{document}